\documentclass[11pt]{article}

\usepackage{amsfonts}
\usepackage{amssymb}
\usepackage{amsthm}
\usepackage{amsmath}
\usepackage{cite}

\newtheorem{Thm}{Theorem}
\newtheorem{Lem}[Thm]{Lemma}
\newtheorem{Cor}[Thm]{Corollary}
\newtheorem{Prop}[Thm]{Proposition}

\theoremstyle{definition}
\newtheorem{Def}{Definition}
\newtheorem{Rem}{Remark}

\DeclareMathOperator{\lmod}{mod}
\DeclareMathOperator{\per}{per}

\begin{document}

\title{Distribution properties of compressing sequences derived from primitive sequences modulo odd prime powers}
\author{Yupeng Jiang and  Dongdai Lin\\\\ \quad State Key Laboratory Of Information Security,\\
Institute of Information Engineering,\\
Chinese Academy of Sciences, Beijing 100093, P.R. China\\ E-mail: \{jiangyupeng,ddlin\}@iie.ac.cn}
\date{}
\maketitle

\begin{abstract}
Let $\underline{a}$ and $\underline{b}$ be primitive sequences over $\mathbb{Z}/(p^e)$ with odd prime $p$
and $e\ge 2$. For certain compressing maps, we consider the distribution properties of compressing sequences
of $\underline{a}$ and $\underline{b}$, and prove that $\underline{a}=\underline{b}$ if the compressing sequences
are equal at the times $t$ such that $\alpha(t)=k$, where $\underline{\alpha}$ is a sequence related to
$\underline{a}$. We also discuss the $s$-uniform distribution property of compressing sequences. For some compressing maps,
we have that there exist different primitive sequences such that the compressing sequences are $s$-uniform. We also
discuss that compressing sequences can be $s$-uniform for how many elements $s$.
\end{abstract}


\section{Introduction}

A sequence over a ring $R$ is denoted by $\underline{a}=(a(t))_{t\ge 0}$ with each $a(t)$ belonging to $R$.
Moreover, if there are elements
$c_0,\ c_1,\ldots,\ c_{n-1}$ in $R$ such that $$a(i+n)=c_{n-1}a(i+n-1)+\cdots+c_1a(i+1)+c_0a(i)$$ holds for all
$i \ge 0$, then the sequence is called a linear recurring sequence of degree $n$ over $R$, generated by
$f(x)=x^n-c_{n-1}x^{n-1}-\cdots-c_1x-c_0$.

Let $p$ be a prime number, $e$ a positive integer and $\mathbb{Z}/(p^e)$ the integer residue
ring modulo $p^e$. We identify the elements of $\mathbb{Z}/(p^e)$ with the corresponding representatives in $\{0,1,2,\ldots, p^e-1\}$.
For two integers $m$ and $n$, the notation $[n]_{\lmod m}$ represents the least nonnegative integer of $n$ modulo $m$
and $[\underline{a}]_{\lmod m}=([a(t)]_{\lmod m})_{t \ge 0}$. These notations were
used by Zheng, Qi and Tian in \cite{ZhQT}. Then a recurring sequence over $\mathbb{Z}/(p^e)$ generated by
$f(x)=x^n-c_{n-1}x^{n-1}-\cdots-c_1x-c_0$ means that for all $i \ge 0$, $a(i) \in \{0,1,2,\ldots, p^e-1\}$ and
$$a(i+n)=[c_{n-1}a(i+n-1)+\cdots+c_1a(i+1)+c_0a(i)]_{\lmod p^e}.$$
Usually, the set of all sequences generated by $f(x)$ over $\mathbb{Z}/(p^e)$ is denoted by $G(f(x),p^e)$.

Let $f(x)$ be a monic polynomial of degree $n$ over $\mathbb{Z}/(p^e)$. If $[f(0)]_{\lmod p}\neq 0$, then there exists
a positive integer $T$ such that $x^T-1$ is divisible by $f(x)$ in $\mathbb{Z}/(p^e)[x]$. The smallest such positive $T$
is called the {\it least period} of $f(x)$ and denoted by $\per(f(x), p^e)$. Ward proved that $\per(f(x), p^e)\le p^{e-1}(p^n-1)$
\cite{Ward}. Polynomials reaching this bound are called {\it primitive} polynomials. A sequence
$\underline{a}$ is called a {\it primitive} sequence of order $n$ if $\underline{a}$ is generated by
a primitive polynomial of degree $n$ and $[\underline{a}]_{\lmod p}$ is not the all zero sequence. It can been shown that primitive sequences
of order $n$ have least period $p^{e-1}(p^n-1)$. For $e=1$, primitive sequences are just the well known m-sequences over
prime field $\mathbb{Z}/(p)$. The set of all primitive sequences generated by primitive polynomial $f(x)$
is usually denoted by $G'(f(x), p^e)$. More details of linear recurring sequences over integer residue rings can be found in \cite{KKMN}.

Let $\underline{a}=(a(t))_{t\ge 0}$ be a sequence over $\mathbb{Z}/(p^e)$. Each $a(t)$ has a unique $p$-adic
expansion as $$a(t)=a_0(t)+a_1(t)\cdot p+\cdots+a_{e-1}(t)\cdot p^{e-1},$$ with $a_i(t) \in \{0,1,\ldots,p-1\}$ for
all $0 \le i \le e-1$. The sequence $\underline{a_i}=(a_{i}(t))_{t\ge 0}$ is called the {\it $i$th-level}
sequence of $\underline{a}$, and
$$\underline{a}=\underline{a_0}+\underline{a_1}\cdot p+\cdots+\underline{a_{e-1}}\cdot p^{e-1}$$
is called the $p$-adic expansion of $\underline{a}$.

Let $f(x)$ be a primitive polynomial over $\mathbb{Z}/(p^e)$ and $\phi(x_0,x_1,\ldots,x_{e-1})$ be an $e$-variable polynomial
over $\mathbb{Z}/(p)$. We can induce a map from the set of primitive sequences $G'(f(x), p^e)$ to the set of sequences over
$\mathbb{Z}/(p)$ by the polynomial $\phi$. The new map is also denoted by $\phi$ and defined by
\begin{align*}
\phi :\quad & G'(f(x), p^e) \rightarrow (\mathbb{Z}/(p))^{\infty} \\
& \underline{a} \mapsto \phi(\underline{a_{0}}, \underline{a_1},\ldots, \underline{a_{e-1}})
=\big(\phi(a_0(t), a_1(t),\ldots, a_{e-1}(t))\big)_{t \ge 0}.
\end{align*}
The map $\phi$ is called a {\it compressing map} and $\phi(\underline{a_{0}}, \underline{a_1},\ldots, \underline{a_{e-1}})$ is called a
{\it compressing sequence}. $\phi(x_0,x_1,\ldots,x_{e-1})$ is called an {\it injective} function if $\phi$ is injective.

Huang and Dai in \cite[Theorem 1]{HD} and Kuzmin and Nechaev in \cite[Theorem 2]{KuNe1993} independently proved that
$\phi(x_0,x_1,\ldots,x_{e-1})=x_{e-1}$ is an injective function. Their result is presented as the following theorem.

\begin{Thm}
Let $f(x)$ be a primitive polynomial of degree $n$ over $\mathbb{Z}/(p^e)$. For $\underline{a}, \underline{b}\in G'(f(x), p^e)$,
$\underline{a}=\underline{b}$ if and only if $\underline{a_{e-1}}=\underline{b_{e-1}}$.
\end{Thm}

The above theorem means that the sequence $\underline{a_{e-1}}$ over $\mathbb{Z}/(p)$ contains all the information of
the sequence $\underline{a}$ over $\mathbb{Z}/(p^e)$. Theoretically, one can recover $\underline{a}$ when given $\underline{a_{e-1}}$.
The cryptographic properties of such compressing sequences have been studied in \cite{KuNe1992,Kuzmin,Dai, DBG}.
From then on, more compressing maps have been proved to be injective. Especially, when $f(x)$ is
a strongly primitive polynomial(Definition \ref{Def:StronglyPrimitive} in the next section),
the following results are obtained for $p=2$ and odd prime $p$ respectively. For $p=2$, Qi, Yang and Zhou \cite{QYZ} proved that almost
all $e$-variable boolean functions containing $x_{e-1}$ are injective. For odd prime $p$, the following theorem has been proved in \cite{ZhuQ2004,ZhuQ2007b,TQ}.

\begin{Thm} \label{Thm:CompressingMap}
Let $f(x)$ be a strongly primitive polynomial of degree $n$ over $\mathbb{Z}/(p^e)$ with odd prime $p$ and $e \ge 2$.
Assume $g(x_{e-1}) \in \mathbb{Z}/(p)[x_{e-1}]$ with $1\le \deg g\le p-1$ and
$\eta(x_0,x_1,\ldots,x_{e-2})\in \mathbb{Z}/(p)[x_0, x_1, \ldots, x_{e-2}]$. Then the function
$$\phi(x_0, x_1, \ldots, x_{e-1})=g(x_{e-1})+\eta(x_0, x_1,\ldots, x_{e-2})$$ is an injective function.
\end{Thm}

The distribution properties of compressing sequences are also interesting. Now we recall some definitions in \cite{ZhQT}.

\begin{Def}
Let $\underline{a}=(a(t))_{t\ge 0}$, $\underline{b}=(b(t))_{t\ge 0}$ and $\underline{c}=(c(t))_{t\ge 0}$
be three sequences over $\mathbb{Z}/(p)$, and let $s,k \in \mathbb{Z}/(p)$. Sequences $\underline{a}$ and
$\underline{b}$ are called $s$-uniform, $s$-uniform with $\underline{c}$ and $s$-uniform with $\underline{c}|_{k}$,
respectively, if $a(t)=s$ iff $b(t)=s$ for all $t \ge 0$, for all $t \ge 0$ with $c(t)\neq 0$ and for all
$t \ge 0$ with $c(t)=k$.
\end{Def}

In \cite{ZhuQ2005} and \cite{ZhuQ2007a}, Zhu and Qi proved that when $e>1$, $\underline{a}=\underline{b}$ iff $\underline{a_{e-1}}$ and
$\underline{b_{e-1}}$ are $0$-uniform. Later, the same authors\cite{ZhuQ2008} obtained a further result that
$\underline{a}=\underline{b}$ iff $\underline{a_{e-1}}$ and
$\underline{b_{e-1}}$ are $0$-uniform with a certain sequence
$\underline{\alpha}$(definition given before Definition \ref{Def:StronglyPrimitive} in the next section).
In \cite{ZhQ2010}, under the condition that $f(x)$ is a strongly primitive polynomial, Zheng and Qi
extended the result to the following more general one. Let $$\phi(x_0, x_1, \ldots, x_{e-1})=x_{e-1}+\eta(x_0, x_1,\ldots, x_{e-2})$$
with the coefficient of $x_{e-2}^{p-1}\cdots x_1^{p-1}x_0^{p-1}$ in $\eta$ is not equal to $(-1)^e\cdot\tfrac{p+1}{2}$.
Then their result is $\underline{a}=\underline{b}$ if and only if there exist $s\in \mathbb{Z}/(p)$ such that
$\phi(\underline{a_{0}}, \underline{a_1},\ldots, \underline{a_{e-1}})$ and
$\phi(\underline{b_{0}}, \underline{b_1},\ldots, \underline{b_{e-1}})$ are $s$-uniform with $\underline{\alpha}$.
Recently, the same author proved a stronger result\cite{ZhQT} which is stated in the following.

\begin{Thm} \label{Thm:s-uniform}
Let $f(x)$ be a strongly primitive polynomial of degree $n$ over $\mathbb{Z}/(p^e)$ with odd prime $p$ and $e\ge 2$.
Assume $$\phi(x_0, x_1, \ldots, x_{e-1})=x_{e-1}+\eta(x_0, x_1,\ldots, x_{e-2})$$
with the coefficient of $x_{e-2}^{p-1}\cdots x_1^{p-1}x_0^{p-1}$ in $\eta$ is not equal to $(-1)^e\cdot\frac{p+1}{2}$. Then
for $\underline{a}, \underline{b}\in G'(f(x), p^e)$, $\underline{a}=\underline{b}$ if and only if there exist
$s\in \mathbb{Z}/(p)$ and $k\in (\mathbb{Z}/(p))^{*}$ such that
$\phi(\underline{a_{0}}, \underline{a_1},\ldots, \underline{a_{e-1}})$ and
$\phi(\underline{b_{0}}, \underline{b_1},\ldots, \underline{b_{e-1}})$ are $s$-uniform with $\underline{\alpha}|_{k}$.
\end{Thm}

In this article, we will investigate that if $x_{e-1}$ is replaced by a general polynomial $g(x_{e-1})$ in the
above theorem, whether or not similar results can still hold. Unfortunately, the answer is negative. We obtain
that if there exists $k \in (\mathbb{Z}/(p))^{*}$ such that
$\phi(a_0(t), a_1(t), \ldots, a_{e-1}(t))=\phi(b_0(t), b_1(t), \ldots, b_{e-1}(t))$
for all $t$ satisfying $\alpha(t)=k$, then $\underline{a}=\underline{b}$.
This result is stronger than Theorem \ref{Thm:CompressingMap}, and we devote Section 3 to prove it. In Section 4,
we consider that for some $g(x_{e-1})$ and different $\underline{a}$ and $\underline{b}$,
$\phi(\underline{a_{0}}, \underline{a_1},\ldots, \underline{a_{e-1}})$ and
$\phi(\underline{b_{0}}, \underline{b_1},\ldots, \underline{b_{e-1}})$ can be $s$-uniform for how many elements in $\mathbb{Z}/(p)$.
In the next section, we recall some facts about primitive sequences modulo odd
prime powers.

\section{Preliminaries}
In this section, we will introduce some facts about sequences over integer residue rings. we only consider
the case that $p$ is an odd prime.

For sequences $\underline{a}=(a(t))_{t\ge 0}$, $\underline{b}=(b(t))_{t \ge 0}$ over $\mathbb{Z}/(p^e)$,
and $c\in \mathbb{Z}/(p^e)$, we have the following operation:
$$\underline{a}+\underline{b}=([a(t)+b(t)]_{\lmod p^e})_{t\ge 0}, \quad c\cdot\underline{a}=([c\cdot a(t)]_{\lmod p^e})_{t\ge 0},$$
$$x^k\underline{a}=(a(t+k))_{t\ge 0}.$$
Then the operation of a polynomial $g(x)=\sum_{k=0}^{n}c_{k}x^k \in \mathbb{Z}/(p^e)[x]$ on the sequence
$\underline{a}$ as
$$g(x)\underline{a}=\sum_{k=0}^{n}c_k\cdot x^k\underline{a}.$$

If $f(x)$ is a primitive polynomial of degree $n$ over $\mathbb{Z}/(p^e)$, it is known\cite{HD} that
there exist polynomials $h_{i}(x)$ of degree less than $n$ over $\mathbb{Z}/(p^e)$ such that for $1\le i\le e$,
\begin{align} \label{eq:period}
 x^{p^{i-1}T}\equiv 1+p^i\cdot h_i(x)\mod {f(x)},
\end{align}
where $T=p^n-1$ and $h_1(x)\equiv h_2(x)\equiv\cdots\equiv h_e(x)\not\equiv 0\mod p$. For a given $f(x)$,
denote by $h_{f}(x)$ the polynomial of $h_1(x)$ modulo $p$. The sequence $\underline{\alpha}$ over
$\mathbb{Z}/(p)$ mentioned in the last section is defined to be $[h_{f}(x)\underline{a_0}]_{\lmod p}$.

\begin{Def} \label{Def:StronglyPrimitive}
Let $f(x)$ be a primitive polynomial and $h_{f}(x)\equiv h_{i}(x)\mod p$. If $h_{f}(x)$ is not constant,
i.e., $\deg (h_{f}(x))\ge 1$, then $f(x)$ is called a strongly primitive polynomial.
\end{Def}

The following definition is given in \cite{TQ}. It is used to deal with carries.

\begin{Def}
For $a=a_0+a_1\cdot p+\cdots+a_{e-1}\cdot p^{e-1}\in \mathbb{Z}/(p^e)$, Define a function
\begin{align*}
C_{1}:\quad \mathbb{Z}/(p^e) &\rightarrow\mathbb{Z}/(p)\\
a &\mapsto C_{1}(a)=a_1.
\end{align*}
For a sequence $\underline{a}=\underline{a_0}+\underline{a_1}\cdot p+\cdots+\underline{a_{e-1}}\cdot p^{e-1}$ over
$\mathbb{Z}/(p^e)$, Define $C_1(\underline{a})=\underline{a_1}$.
\end{Def}

For an element $u$ in $\mathbb{Z}/(p)=\{0,1,\ldots,p-1\}$, the function $C_1$ can induce a map from $\mathbb{Z}/(p)$ to itself by
$x\mapsto C_1(u+x)$. By Lagrange interpolation, each such function has a unique polynomial representation,
and we have the following result\cite{ZhQ2010}.

\begin{Prop}\label{Prop:carry}
The coefficient of $x^{p-1}$ in the polynomial representation  of the map $x \mapsto C_1(u+x)$ is $-u$.
\end{Prop}

Let the equality \eqref{eq:period} operates on a sequence $\underline{a}$ generated by $f(x)$ over $\mathbb{Z}/(p^e)$.
The following results can be proved. For details see \cite{TQ} and \cite{ZhQ2010}.
\begin{Prop}\label{Prop:recurring}
Let $f(x)$ be a primitive polynomial of degree $n$ over $\mathbb{Z}/(p^e)$. Assume
$\underline{a}\in G(f(x), p^e)$, $T=p^n-1$, $h_{i}(x)$ and $\underline{\alpha}$ are defined as above.
Then for integers $j \ge 0$, we have the following results.\\
{\rm (1)} The equality
\begin{align} \label{eq:e ge 2}
(x^{j\cdot p^{e-2}T}-1)\underline{a_{e-1}}\equiv j\cdot\underline{\alpha}\mod p
\end{align}
 holds for $e \ge 2$.\\
{\rm (2)} The equality
\begin{align} \label{eq:e ge 4}
(x^{j\cdot p^{e-3}T}-1)\underline{a_{e-1}}\equiv & j\cdot (h_{f}(x)\underline{a_1})+C_{1}(j\cdot (h_{e-2}(x)\underline{a_0})+\\ \notag
&C_{1}(\underline{a_{e-2}}+[j\cdot \underline{\alpha}]_{\lmod p}) \mod p
\end{align}
holds for $e\ge 4$.\\
{\rm (3)} The equality
\begin{align} \label{eq:e=3}
(x^{j\cdot T}-1)\underline{a_{2}} \equiv &{j \choose 2}h_f^2(x)\underline{a_0}+j\cdot (h_{f}(x)\underline{a_1})+\\ \notag
&C_{1}(j\cdot (h_{1}(x)\underline{a_0})+C_{1}(\underline{a_{1}}+[j\cdot \underline{\alpha}]_{\lmod p}) \mod p
\end{align}
holds.
\end{Prop}

The following statements about periods of linear recurring sequences over $\mathbb{Z}/(p^e)$ can be proved similarly as in \cite{Dai}.
\begin{Prop} \label{Prop:Period}
Let $f(x)$ be a primitive polynomial of degree $n$ over $\mathbb{Z}/(p^e)$ and $\underline{a}\in G(f(x), p^e)$
have $p$-adic expansion $\underline{a}=\underline{a_0}+\underline{a_1}\cdot p+\cdots+\underline{a_{e-1}}\cdot p^{e-1}$.
Assume $T=p^n-1$. Then\\
{\rm (1)} if $\underline{a_0}\neq \underline{0}$, then $\per(\underline{a_i})=p^iT$ for $0\le i \le e-1$ and
$\per(\underline{a})=p^{e-1}T$,  \\
{\rm (2)} if $\underline{a_0}=\underline{a_1}=\cdots=\underline{a_{i-1}}=\underline{0}$ and
$\underline{a_i} \neq \underline{0}$ for $0\le i \le e-1$,
then $\per(\underline{a})=p^{e-1-i}T$.
\end{Prop}

For m-sequences over $\mathbb{Z}/(p)$, The following results are well known.

\begin{Prop} \label{Prop:LinearRelation}
Let $f(x)$ be a primitive polynomial of degree $n$ over $\mathbb{Z}/(p)$, and $\underline{a}\in G(f(x), p)$,
$\underline{b} \in G'(f(x), p)$. Then \\
{\rm (1)} if $\underline{a}$ and $\underline{b}$ are linearly dependent over $\mathbb{Z}/(p)$ with
$\underline{a}=\lambda\cdot \underline{b}$, then the set $\{a(t)\mid b(t)=k\}$ is equal to $\{\lambda k\}$,\\
{\rm (2)} if $\underline{a}$ and $\underline{b}$ are linearly independent over $\mathbb{Z}/(p)$,
then $\{a(t)\mid b(t)=k\}=\{0, 1, \ldots, p-1\}$.
\end{Prop}

\section{Distribution at $\alpha(t)=k$}
In this section we always assume $p$ is an odd prime and $e \ge 2$.
Let $f(x)$ be a strongly primitive polynomial over $\mathbb{Z}/(p^e)$,
$\underline{a}, \underline{b}\in G'(f(x), p^e)$ and $\underline{\alpha}=[h_{f}(x)\underline{a_0}]_{\lmod p}$.
Assume $g(x_{e-1})\in\mathbb{Z}/(p)[x_{e-1}]$ and
$\eta_{e-2}(x_0, x_1, \ldots, x_{e-2}) \in \mathbb{Z}/(p)[x_0, x_1, \ldots, x_{e-1}]$. Let
$$\phi(x_0, x_1, \ldots, x_{e-1})=g(x_{e-1})+\eta_{e-2}(x_0, x_1, \ldots, x_{e-2}).$$
We will prove that for sequences $\underline{a}$, $\underline{b}\in G'(f(x), p^e)$,
$\underline{a}=\underline{b}$ if and only if there exists some $k \in (\mathbb{Z}/(p))^*$ such that
the compressing sequences
$\phi(a_{0}(t), a_1(t),\ldots, a_{e-1}(t))=
\phi(b_{0}(t), b_1(t),\ldots, b_{e-1}(t))$ at $t$ with
$\alpha(t)=k$.
We depart the proof into two cases: ($1$) $\deg g=1$ and ($2$) $2 \le \deg g \le p-1$.

\subsection{Case $\deg g=1$}
When $\deg g=1$, we can assume $g(x_{e-1})=x_{e-1}$ without loss of generality. In this case, we do not
need $f(x)$ to be a strongly primitive polynomial.  In the following, most equalities are regarded as
over $\mathbb{Z}/(p)$. we have
\begin{Thm}\label{Thm:deg=1}
Let $e \ge 2$ and $f(x)$ be a primitive polynomial of degree $n$ over $\mathbb{Z}/(p^e)$.
Assume $\underline{a},\ \underline{b}\in G'(f(x), p^e)$ and
$$\phi(x_0, x_1, \ldots, x_{e-1})=x_{e-1}+\eta_{e-2}(x_0, x_1, \ldots, x_{e-2}),$$
where $\eta_{e-2}$ is an $(e-1)$-variable polynomial over $\mathbb{Z}/(p)$. Then
$\underline{a}=\underline{b}$ if and only if there exists $k \in (\mathbb{Z}/(p))^*$ such that the compressing sequences
$\phi(a_{0}(t), a_1(t),\ldots, a_{e-1}(t))=
\phi(b_{0}(t), b_1(t), \ldots, b_{e-1}(t))$ at $t$ with
$\alpha(t)=k$.
\end{Thm}

\begin{proof}
Suppose $\alpha(t)=k\neq 0$, we have
\begin{align} \label{eq:main1}
a_{e-1}(t)+\eta_{e-2}(a_0(t), \ldots, a_{e-2}(t))=b_{e-1}(t)+\eta_{e-2}(b_0(t), \ldots, b_{e-2}(t)).
\end{align}
Let $T=p^n-1$. Then for all $j \ge 0$, $\alpha(t+j\cdot p^{e-2}T)=\alpha(t)=k$.
Since $p^{e-2}T$ is a period of $\underline{a_{i}}, \underline{b_j}$ for $i, j< e-1$, then if we replace $t$
by $j\cdot p^{e-2}T+t$ in \eqref{eq:main1} and minus it we have
$$a_{e-1}(j\cdot p^{e-2}T+t)-a_{e-1}(t)=b_{e-1}(j\cdot p^{e-2}T+t)-b_{e-1}(t).$$
Let $\underline{\beta}=[h_{f}(x)\underline{b_0}]_{\lmod p}$.
By \eqref{eq:e ge 2} in Proposition \ref{Prop:recurring}, we obtain $j\cdot \alpha(t)=j \cdot \beta(t)$ which
means that $\beta(t)=k$ when $\alpha(t)=k$. As $\underline{\alpha}, \underline{\beta} \in G'(f(x),p)$,
by Proposition \ref{Prop:LinearRelation}, we have $\underline{\alpha}=\underline{\beta}$, and then
$\underline{a_0}=\underline{b_0}$.

When $e \ge 4$, the equality \eqref{eq:e ge 4} in Proposition \ref{Prop:recurring} is just
\begin{align*}
a_{e-1}&(j\cdot p^{e-3}T+t)-a_{e-1}(t)= j\cdot (h_{f}(x)a_1)(t)+\\
&C_{1}(j\cdot (h_{e-2}(x)a_0)(t))+C_{1}(a_{e-2}(t)+[j\cdot \alpha(t)]_{\lmod p})
\end{align*}
Since the same equality holds also for $\underline{b_{e-1}}$ and
$\underline{a_0}=\underline{b_0}$, we have
\begin{align*}
j\cdot (h_{f}(x)a_1)(t)+C_{1}(a_{e-2}(t)+[j\cdot \alpha(t)]_{\lmod p})+&\\
\eta_{e-2}(a_0(t), \ldots, a_{e-2}(j\cdot p^{e-3}T+t))-\eta_{e-2}(a_0(t), \ldots, a_{e-2}(t))&=\\
\qquad j\cdot (h_{f}(x)b_1)(t)+C_{1}(b_{e-2}(t)+[j\cdot \beta(t)]_{\lmod p})+&\\
\eta_{e-2}(b_0(t), \ldots, b_{e-2}(j\cdot p^{e-3}T+t))-\eta_{e-2}(b_0(t), \ldots, b_{e-2}(t)).&
\end{align*}
As $\underline{a_0}=\underline{b_0}$, by \eqref{eq:e=3}, the above equation also holds for $e=3$.
Let $\underline{\tau}=h_{f}(x)(\underline{a_1}-\underline{b_1})$. From $\alpha(t)=\beta(t)=k$ and
replacing $e$ by $e-1$ in \eqref{eq:e ge 2}, we have
$$a_{e-2}(j\cdot p^{e-3}T+t)=a_{e-2}(t)+j\cdot k,$$
$$b_{e-2}(j\cdot p^{e-3}T+t)=b_{e-2}(t)+j\cdot k,$$
and then
\begin{align*}
j\cdot \tau(t)+C_{1}(a_{e-2}(t)+[j\cdot k]_{\lmod p})-C_{1}(b_{e-2}(t)+[j\cdot k]_{\lmod p})=\\
\eta_{e-2}(b_0(t), \ldots, b_{e-2}(t)+j\cdot k)-\eta_{e-2}(b_0(t), \ldots, b_{e-2}(t))-\\
\eta_{e-2}(a_0(t), \ldots, a_{e-2}(t)+j\cdot k)+\eta_{e-2}(a_0(t), \ldots, a_{e-2}(t))
\end{align*}
As $k\neq 0$, when $j$ runs over $\{0,1,\ldots, p-1\}$, $[j\cdot k]_{\lmod p}$ also runs over
$\{0,1,\ldots, p-1\}$. Let $[j\cdot k]_{\lmod p}=x$ and then $j=k^{-1}x$. Each side of the above equality can
be regarded as a function from $\mathbb{Z}/(p)$ to itself. As such a function can be represented by a
polynomial, we denote by $L(x)$ for the left side function and $R(x)$ for the right side. Then we have
$$L(x)=k^{-1}\tau(t)x+C_{1}(a_{e-2}(t)+x)-C_{1}(b_{e-2}(t)+x).$$
By Proposition \ref{Prop:carry}, the coefficient of $x^{p-1}$ in $L(x)$ is $-a_{e-2}(t)+b_{e-2}(t)$. Assume
$$\eta_{e-2}(x_0, \ldots, x_{e-2})=-\eta_{e-3}(x_0, \ldots, x_{e-3})\cdot x_{e-2}^{p-1}+\rho_{e-2}(x_0,\ldots, x_{e-2}),$$
where the degree of $x_{e-2}$ in $\rho_{e-2}(x_0, \ldots, x_{e-2})$ is less than $p-1$. Then the coefficient
of $x^{e-1}$ in $R(x)$ is $\eta_{e-3}(a_0(t), \ldots, a_{e-3}(t))-\eta_{e-3}(b_0(t), \ldots, b_{e-3}(t))$.
Now we have
$$-a_{e-2}(t)+b_{e-2}(t)=\eta_{e-3}(a_0(t), \ldots, a_{e-3}(t))-\eta_{e-3}(b_0(t), \ldots, b_{e-3}(t)).$$
and then
$$a_{e-2}(t)+\eta_{e-3}(a_0(t), \ldots, a_{e-3}(t))=b_{e-2}(t)+\eta_{e-3}(b_0(t), \ldots, b_{e-3}(t)),$$
which reduce $e-1$ in \eqref{eq:main1} to $e-2$.

By induction, we have that when $\alpha(t)=k$,
$$a_{i}(t)+\eta_{i-1}(a_0(t), \ldots, a_{i-1}(t))=b_{i}(t)+\eta_{i-1}(b_0(t), \ldots, b_{i-1}(t))$$
holds for $i=1, 2, \ldots, e-1$, where $\eta_{i-1}(x_0, \ldots, x_{i-1})$ is an $i$-variable polynomial.
Let $\underline{c}=\underline{a}-\underline{b}$. For $i=1$, since $\underline{a_0}=\underline{b_0}$,
we have $\underline{c_1} \in G(f(x), p)$ and $c_1(t)=a_1(t)-b_1(t)=0$ whenever $\alpha(t)=k$. Then by
Proposition \ref{Prop:LinearRelation}, we have $\underline{c_1}=\underline{0}$, i.e.,
$[\underline{a}]_{\mod p^2}=[\underline{b}]_{\lmod p^2}$. If we have proved that
$[\underline{a}]_{\lmod p^m}=[\underline{b}]_{\lmod p^m}$ for some $m<e$, then we have
$\underline{c_{m}}=\underline{a_m}-\underline{b_{m}} \in G(f(x), p)$
and $c_m(t)=0$ whenever $\alpha(t)=k$. Then $\underline{c_m}=\underline{0}$ and
$[\underline{a}]_{\lmod p^{m+1}}=[\underline{b}]_{\lmod p^{m+1}}$. By induction, we can finally prove
that $[\underline{a}]_{\lmod p^e}=[\underline{b}]_{\lmod p^e}$, i.e., $\underline{a}=\underline{b}$.
\end{proof}

\subsection{Case $2\le \deg g\le p-1$}
In this subsection, we will prove the result for $2\le \deg g\le p-1$. We prove the following
lemmas first.

\begin{Lem} \label{Lem:relation}
Let $f(x)$ be a primitive polynomial of degree $n$ over $\mathbb{Z}/(p^e)$ and $\underline{c}\in G(f(x), p^e)$.
Assume $\underline{\gamma}\in G'(f(x), p)$. For $k\in (\mathbb{Z}/(p))^*$,
the set $\{c_{e-1}(t)\mid \gamma(t)=k\}$ runs over all elements in $\mathbb{Z}/(p)$ or is
a singleton. Moreover, the latter case happens only if
$\underline{c_0}=\cdots=\underline{c_{e-2}}=\underline{0}$
and $\underline{c_{e-1}}=\lambda\cdot \underline{\gamma}$, and the singleton is $\{\lambda\cdot k\}$.
\end{Lem}

\begin{proof}
If $\underline{c_0}=\cdots=\underline{c_{j-1}}=\underline{0}$ and $\underline{c_{j}}\neq \underline{0}$
with $0 \le j \le e-2$, then $\underline{c}=p^j\underline{c'}$ and $\underline{c'} \in G'(f(x), p^{e-j})$. Since
$\underline{c_0'} \in G'(f(x),p)$, let $h_{f}(x)$ be defined as in Definition \ref{Def:StronglyPrimitive} of last section,
then $h_f(x)\underline{c'_{0}}\in G'(f(x), p)$. By Proposition \ref{Prop:LinearRelation},
for some $t$ with $\gamma(t)=k \neq 0$, we have $(h_{f}(x)c'_0)(t)\neq 0$.
Replacing $e$ by $e-j$ in equality \eqref{eq:e ge 2}, then $c'_{e-j-1}(t)$ can be any element in $\mathbb{Z}/(p)$.
As $\underline{c_{e-1}}=\underline{c'_{e-j-1}}$, then $\underline{c_{e-1}}$ can be any element in $\mathbb{Z}/(p)$.

Assume $\underline{c_0}=\cdots=\underline{c_{e-2}}=\underline{0}$. Then $\underline{c_{e-1}}\in G(f(x), p)$.
By Proposition \ref{Prop:LinearRelation},
if $\underline{c_{e-1}}$ and $\underline{\gamma}$ are linearly independent, then
$\underline{c_{e-1}}$ can be any element in $\mathbb{Z}/(p)$. If $\underline{c_{e-1}}$ and
$\underline{\gamma}$ are linearly dependent with $\underline{c_{e-1}}=\lambda\cdot\underline{\gamma}$,
then it is obvious that the set
$\{\underline{c_{e-1}(t)}\mid \gamma(t)=k\}$ is the singleton $\{\lambda\cdot k\}$.
\end{proof}

\begin{Lem} \label{Lem:highestlevel}
Let $f(x)$ be a strongly primitive polynomial of degree $n$ over $\mathbb{Z}/(p^e)$ with odd prime $p$
and $e\ge 2$. Assume $\underline{a},\ \underline{b}\in G'(f(x), p^e)$,
$\underline{\alpha}=[h_{f}(x)\underline{a_0}]_{\lmod p}$ and
$\underline{\beta}=[h_{f}(x)\underline{b_0}]_{\lmod p}$. Suppose $\underline{\beta}=\lambda\cdot\underline{\alpha}$
holds for some $\lambda\in \{1, 2, \ldots, p-1\}$. Let $k\in (\mathbb{Z}/(p))^*$.
If for those $t$ with $\alpha(t)=k$, the equality
$$b_{e-1}(t)=\delta+\lambda\cdot a_{e-1}(t)$$
always holds, then $\lambda=1$, $[\underline{a}]_{\lmod p^{e-1}}=[\underline{b}]_{\lmod p^{e-1}}$ and the
sequence $\underline{b_{e-1}}-\underline{a_{e-1}}=\delta k^{-1}\cdot \underline{\alpha}$.
\end{Lem}

\begin{proof}
As $f(x)$ is a strongly primitive polynomial, $\underline{\alpha}$ and $\underline{a_0}$ are linearly independent.
By Proposition \ref{Prop:LinearRelation}, when $\alpha(t)=k$, $a_0(t)$ can be any element in $\{0, 1, \ldots, p-1\}$.
Applying the equality \eqref{eq:e ge 2} for $\underline{a_1}, \ldots, \underline{a_{e-1}}$,
then when $\alpha(t)=k$, $a(t)$ can be any element in $\{0, 1, \ldots, p^e-1\}$.

If $1\le \lambda <p-1$, let $\underline{c}=\underline{b}-\lambda\cdot\underline{a}\in G(f(x), p^e)$. Then
$$c_{e-1}(t)=b_{e-1}(t)-\lambda\cdot a_{e-1}(t)-u(t)=\delta-u(t)$$
with $u(t)$ satisfying
$$[b(t)]_{\lmod p^{e-1}}+(u(t)-1)p^{e-1}<\lambda\cdot [a(t)]_{\lmod p^{e-1}}\le [b(t)]_{\lmod p^{e-1}}+u(t)p^{e-1}.$$
Since $0 \le [a(t)]_{\lmod p^{e-1}}, [b(t)]_{\lmod p^{e-1}} <p^{e-1}$, we have $0 \le u(t) \le \lambda$.
Since $u(t)\le \lambda<p-1$, $c_{e-1}(t)=\delta-u(t)$ can not be all elements in
$\{0, 1, \ldots, p-1\}$. By Lemma \ref{Lem:relation}, $c_{e-1}(t)$ must be a constant when $\alpha(t)=k$.
We choose $t$ with $a(t)=0$, then $u(t)=0$. When $\lambda \ge 2$, we can also choose $t$ with
$\lambda\cdot [a(t)]_{\lmod p^{e-1}}>p^{e-1}$, then $u(t)>0$. So $c_{e-1}(t)$ is not a constant,
which is contradiction to Lemma \ref{Lem:relation}.
If $\lambda=1$, then $u(t)$ can only be $0$ or $1$ and $0$ is reachable when
$[a(t)]_{\lmod p^{e-1}}=0$. By Lemma \ref{Lem:relation}, we
have $u(t)$ can only be $0$, and then $\underline{c_0}=\cdots=\underline{c_{e-2}}=\underline{0}$ and
$c_{e-1}(t)=\delta$ when $\alpha(t)=k$.

If $\lambda=p-1$, let $\underline{c}=\underline{a}+\underline{b} \in G(f(x), p^e)$. Then
$$c_{e-1}(t)=b_{e-1}(t)+a_{e-1}(t)+u(t)=\delta+u(t)$$
with $u(t)$ satisfying
\begin{align*}
u(t)=\begin{cases}
0 & [a(t)]_{\lmod p^{e-1}}+[b(t)]_{\lmod p^{e-1}}<p^{e-1},\\
1 & \text{otherwise}.
\end{cases}
\end{align*}
First we choose $[a(t)]_{\lmod p^{e-1}}=0$. As $[b(t)]_{\lmod p^{e-1}}<p^{e-1}$, we have $u(t)=0$.
Then we choose $[a(t)]_{\lmod p^{e-1}}=p^{e-1}-1$, as $\underline{\beta}=\lambda\cdot\underline{\alpha}$
with $\lambda\neq 0$, we have $\underline{b_0}=\lambda\cdot\underline{a_0}$. Then $b_0(t) \neq 0$ and
$[a(t)]_{\lmod p^{e-1}}+[b(t)]_{\lmod p^{e-1}} \ge p^{e-1}$. Thus $u(t)=1$ and $c_{e-1}(t)$ can only
choose $2$ elements $\{\delta, \delta+1\}$, which is contradiction to Lemma \ref{Lem:relation}.

From the above discussion, we have $\lambda=1$, and then $\underline{c}=\underline{a}-\underline{b}$ and
$[\underline{c}]_{p^{e-1}}=\underline{0}$, which means
$[\underline{a}]_{\lmod p^{e-1}}=[\underline{b}]_{\lmod p^{e-1}}$.
The statement $c_{e-1}(t)=\delta$ when $\alpha(t)=k$ means
$\underline{b_{e-1}}-\underline{a_{e-1}}=\delta k^{-1}\cdot \underline{\alpha}$
by Proposition \ref{Prop:LinearRelation}.
The proof is complete.
\end{proof}

\begin{Thm}\label{Thm:degge2}
Let $e \ge 2$ and $f(x)$ be a strongly primitive polynomial of degree $n$ over $\mathbb{Z}/(p^e)$.
Assume $\underline{a},\ \underline{b}\in G'(f(x), p^e)$ and
$$\phi(x_0, x_1, \ldots, x_{e-1})=g(x_{e-1})+\eta_{e-2}(x_0, x_1, \ldots, x_{e-2}),$$
where $2 \le \deg g \le p-1$ and $\eta_{e-2}$ is an $(e-1)$-variable polynomial over $\mathbb{Z}/(p)$. Then
$\underline{a}=\underline{b}$ if and only if there exists some $k \in (\mathbb{Z}/(p))^*$ such that
the compressing sequences
$\phi(a_{0}(t), a_1(t),\ldots, a_{e-1}(t))=
\phi(b_{0}(t), b_1(t), \ldots, b_{e-1}(t))$ at $t$ with
$\alpha(t)=k$.
\end{Thm}

\begin{proof}
Suppose $\alpha(t)=k\neq 0$, we have
\begin{align} \label{eq:main2}
g(a_{e-1}(t))+\eta_{e-2}(a_0(t), \ldots, a_{e-2}(t))=g(b_{e-1}(t))+\eta_{e-2}(b_0(t), \ldots, b_{e-2}(t)).
\end{align}
Let $T=p^n-1$. Then for all $j \ge 0$, $\alpha(t+j\cdot p^{e-2}T)=\alpha(t)=k$.
Since $p^{e-2}T$ is a period of $\underline{a_{i}}, \underline{b_j}$ for $i, j< e-1$, then if we replace $t$
by $j\cdot p^{e-2}T+t$ in \eqref{eq:main2} and minus it we have
\begin{align}\label{eq:degge2}
g(a_{e-1}(j\cdot p^{e-2}T+t))-g(a_{e-1}(t))=g(b_{e-1}(j\cdot p^{e-2}T+t))-g(b_{e-1}(t)).
\end{align}
Let $\underline{\beta}=[h_{f}(x)\underline{b_0}]_{\lmod p}$.
Again by \eqref{eq:e ge 2} in Proposition \ref{Prop:recurring}, we have
\begin{align}
a_{e-1}(j\cdot p^{e-2}T+t)=a_{e-1}(t)+j\cdot \alpha(t),\label{eq:a}\\
b_{e-1}(j\cdot p^{e-2}T+t)=b_{e-1}(t)+j\cdot \beta(t). \label{eq:b}
\end{align}
If there exists $t$ such that $\alpha(t)=k$ and $\beta(t)=0$, substituting the above two equalities to
\eqref{eq:degge2} we obtain that
$$g(a_{e-1}(t)+j\cdot k)-g(a_{e-1}(t))=0$$
holds for all $j$. Thus $g$ must be a constant polynomial and it is a contradiction to $2\le \deg g \le p-1$.
So when $\alpha(t)=k$, $\beta(t)$ can not be zero. By Proposition \ref{Prop:LinearRelation}, we have
$\underline{\beta}=\lambda\cdot \underline{\alpha}$ with $\lambda\neq 0$.
By \eqref{eq:a}, we can choose $t$ such that $\alpha(t)=k$ and $a_{e-1}(t)=0$. If for this $t$
we have $b_{e-1}(t)=\delta$, then from \eqref{eq:degge2}, \eqref{eq:a} and \eqref{eq:b}, we have
$$g(j\cdot k)-g(0)=g(\delta+j\cdot\lambda k)-g(\delta).$$
We are going to prove that for $t$ with $\alpha(t)=k$ and $a_{e-1}(t)=0$, we always have $b_{e-1}(t)=\delta$.
If $b_{e-1}(t)=\delta'\neq \delta$, then we also have
$$g(j\cdot k)-g(0)=g(\delta'+j\cdot\lambda k)-g(\delta').$$
From the above two equalities, $$g(\delta+j\cdot\lambda k)-g(\delta)=g(\delta'+j\cdot\lambda k)-g(\delta').$$
As $j\cdot\lambda k$ runs over all elements in $\mathbb{Z}/(p)$, we have
$g(\delta+x)-g(\delta'+x)$ is equal to a constant $g(\delta)-g(\delta')$. Since
$\deg g\ge 2$, it is impossible. Then we have proved that for $t$ with $\alpha(t)=k$ and
$a_{e-1}(t)=0$, $b_{e-1}(t)$ is equal to $\delta$. From $\underline{\beta}=\lambda\cdot \underline{\alpha}$
and equalities \eqref{eq:a} and \eqref{eq:b}, we have for $t$ with $\alpha(t)=k$ and $a_{e-1}(t)=0$,
$$b_{e-1}(j\cdot p^{e-2}T+t)=\delta+\lambda\cdot a_{e-1}(j\cdot p^{e-2}T+t).$$
The above equality holds for all $j\ge 0$. Then for $t_0$ with $\alpha(t_0)=k$, there exists
some $0\le j_0\le p-1$ such that $a_{e-1}(j_0\cdot p^{e-2}T+t_0)=0$. The above equation holds for
$t=j_0\cdot p^{e-2}T+t_0$ and $j=p-j_0$. then
\begin{align*}
b_{e-1}(t_0)&=b_{e-1}(p^{e-1}T+t_0)=b_{e-1}(j\cdot p^{e-2}T+t)\\
&=\delta+\lambda\cdot a_{e-1}(j\cdot p^{e-2}T+t)\\
&=\delta+\lambda\cdot a_{e-1}(p^{e-1}T+t_0)\\
&=\delta+\lambda\cdot a_{e-1}(t_0).
\end{align*}
We have proved that for $t$ with $\alpha(t)=k$, $b_{e-1}(t)=\delta+\lambda\cdot a_{e-1}(t)$ always holds.
By Lemma \ref{Lem:highestlevel}, we have
$$\lambda=1,\quad [\underline{a}]_{\lmod p^{e-1}}=[\underline{b}]_{\lmod p^{e-1}},\quad \underline{b_{e-1}}=\underline{a_{e-1}}+\delta.$$
Now we go back to equality \eqref{eq:main2}. Since when $\alpha(t)=k$, $a_{e-1}(t)$ can be any element in
$\{0, 1, \ldots, p-1\}$, then we have $g(x)=g(x+\delta)$. As $2 \le \deg g\le p-1$, we must have $\delta=0$.
Thus $\underline{b_{e-1}}=\underline{a_{e-1}}$ and we have proved that
$[\underline{a}]_{\lmod p^{e-1}}=[\underline{b}]_{\lmod p^{e-1}}$. So $\underline{a}=\underline{b}$, and the
proof is complete.
\end{proof}

From Theorem \ref{Thm:deg=1} and Theorem \ref{Thm:degge2}, we have proved what we state at the
beginning of this section.

\section{$s$-uniform}
Let $f(x)$ be a strongly primitive polynomial of degree $n$ over $\mathbb{Z}/(p^e)$ with
odd prime $p$ and $e\ge 2$. Assume $\underline{a}, \underline{b}\in G'(f(x), p^e)$. Let
$$\phi(x_0, x_1, \ldots, x_{e-1})=g(x_{e-1})+\eta_{e-2}(x_0, x_1, \ldots, x_{e-2})$$
with $\eta_{e-2}$ an $(e-1)$-variable polynomial over $\mathbb{Z}/(p)$.
In Theorem \ref{Thm:s-uniform}\cite[Theorem 9]{ZhQT}, Zheng, Qi and Tian prove that
when $g(x_{e-1})=x_{e-1}$ and the coefficient of
$x_{e-2}^{p-1}\cdots x_1^{p-1}x_0^{p-1}$ in $\eta_{e-2}$ is not equal to $(-1)^e\cdot\frac{p+1}{2}$,
then $\underline{a}=\underline{b}$ if and only if there exist
$s\in \mathbb{Z}/(p)$ and $k\in (\mathbb{Z}/(p))^{*}$ such that
$\phi(\underline{a_{0}}, \underline{a_1},\ldots, \underline{a_{e-1}})$ and
$\phi(\underline{b_{0}}, \underline{b_1},\ldots, \underline{b_{e-1}})$ are $s$-uniform with $\underline{\alpha}|_{k}$.
In this section, we will discuss the $s$-uniform property for general $g(x_{e-1})$.

If the image of polynomial $\phi$ is not $\mathbb{Z}/(p)$, i.e., there is some $s$ in $\mathbb{Z}/(p)$ such
that $\phi(x_0, x_1,\ldots, x_{e-1})\neq s$ for all $e$-tuples in $(\mathbb{Z}/(p))^e$ , then it is obvious
that for any $\underline{a}$ and $\underline{b}$, $\phi(a_{0}(t), a_1(t),\ldots, a_{e-1}(t))=s$
if and only if $\phi(b_{0}(t), b_1(t), \ldots, b_{e-1}(t))=s$.
Thus we only consider the case that $s$ is an image of $\phi$. Since
when $\alpha(t)=k$, $a(t)$ can be any element in $\{0,1, \ldots, p^e-1\}$ by the proof of Lemma \ref{Lem:highestlevel},
$s$ is an image of $\phi$ if and only if there is some $t$ with $\alpha(t)=k$ such that
$\phi(a_0(t), a_1(t), \ldots, a_{e-1}(t))=s$.

We first consider the case that $g(x_{e-1})$ is a permutation polynomial. We define a function
$\psi_{z, w}$ from $(\mathbb{Z}/(p))^{e-1}$ to $\mathbb{Z}/(p)$ by
\begin{align*}
\psi_{z, w}(x_0, \ldots, x_{e-2})=\begin{cases}
z, & \text{if } (x_0, \ldots, x_{e-2})=(0, \ldots, 0),\\
w, & \text{if } (x_0, \ldots, x_{e-2})\neq (0, \ldots, 0).
\end{cases}
\end{align*}
The polynomial representation of $\psi_{z, w}$ is
$$\psi_{z, w}(x_0, \ldots, x_{e-2})=(z-w)(1-x_0^{p-1})\cdots(1-x_{e-2}^{p-1})+w.$$

We have the following theorem.
\begin{Thm}
Let $\underline{a},\underline{b}\in G'(f(x), p^e)$ and $\underline{b}=-\underline{a}$.
Assume
$$\phi(x_0, \ldots, x_{e-1})=g(x_{e-1})+\psi_{z, w}(x_0, \ldots, x_{e-2})$$
with $g(x_{e-1})$ a permutation polynomial.
Then, for any $s$ in $\mathbb{Z}/(p)$, we can choose suitable $z$ and $w$ such that
$\phi(\underline{a_{0}}, \underline{a_1},\ldots, \underline{a_{e-1}})$ and
$\phi(\underline{b_{0}}, \underline{b_1},\ldots, \underline{b_{e-1}})$ are $s$-uniform.
\end{Thm}
\begin{proof}
We solve the equations
\[\left\{
\begin{array}{l}
g(\frac{p-1}{2})+w=s \\
g(0)+z=s
\end{array}
\right.\]
to get the unique $z$ and $w$. As $g(x_{e-1})$ is a permutation polynomial and $\psi_{z,w}$ is
a two-value function. $\phi(x_0, \ldots, x_{e-1})=s$ if and only if the following two cases happen
\begin{align*}
&g(x_{e-1})=s-z\qquad\psi_{z, w}(x_0, \ldots, x_{e-2})=z,\\
&g(x_{e-1})=s-w\qquad\psi_{z, w}(x_0, \ldots, x_{e-2})=w.
\end{align*}
If and only if $(x_0, \ldots, x_{e-2}, x_{e-1})$ satisfies one of the following two conditions
\begin{equation*}
\begin{split}
&(x_0, \ldots, x_{e-2})=(0, \ldots, 0)\quad x_{e-1}=0, \\
&(x_0, \ldots, x_{e-2})\neq(0, \ldots, 0)\quad x_{e-1}=\frac{p-1}{2}.
\end{split}\tag{*}
\end{equation*}
When $\underline{b}=-\underline{a}$, the $e$-tuple $\big(a_{0}(t), \ldots a_{e-2}(t), a_{e-1}(t)\big)$ satisfies ($*$) if
and only if $\big(b_{0}(t), \ldots b_{e-2}(t), b_{e-1}(t)\big)$ satisfies ($*$). Then for time $t$, we have
$\phi(a_{0}(t), a_1(t),\ldots, a_{e-1}(t))=s$
if and only if $\phi(b_{0}(t), b_1(t), \ldots, b_{e-1}(t))=s$, i.e.,
$\phi(\underline{a_{0}}, \underline{a_1},\ldots, \underline{a_{e-1}})$ and
$\phi(\underline{b_{0}}, \underline{b_1},\ldots, \underline{b_{e-1}})$ are $s$-uniform.
\end{proof}

\begin{Rem}
In \cite[Theorem 21]{ZhQT}, Zheng, Qi and Tian give a counterexample of Theorem \ref{Thm:s-uniform},
if the condition that the coefficient of $x_{e-2}^{p-1}\cdots x_1^{p-1}x_0^{p-1}$ in
$\eta_{e-2}$ is equal to $(-1)^e\cdot\frac{p+1}{2}$. Their result is when
$$\phi(x_0, \ldots, x_{e-1})=x_{e-1}+(-1)^e(x_{e-2}^{p-1}-1)\cdots(x_0^{p-1}-1)-\frac{p-1}{2},$$
for $\underline{b}=-\underline{a}$, the corresponding compressing sequences are $0$-uniform.
If we let $g(x_{e-1})=x_{e-1}$ and $s=0$ in the above theorem, then $z=0$ and $w=\tfrac{p+1}{2}$.
From the polynomial representation of $\psi_{z, w}$, our result coincides with their's.
\end{Rem}

When $g(x_{e-1})$ is not a permutation polynomial and satisfies an additional condition in the following theorem,
there exist many choices of $\eta_{e-2}$ such that for different $\underline{a}, \underline{b}\in G'(f(x), p^e)$,
$\phi(\underline{a_{0}}, \underline{a_1},\ldots, \underline{a_{e-1}})$ and
$\phi(\underline{b_{0}}, \underline{b_1},\ldots, \underline{b_{e-1}})$ are $s$-uniform. For a nonempty subset $W$ of
$\{0, 1, \ldots, p-1\}$, we use $\psi_{z, W}$ to denote any function from $(\mathbb{Z}/(p))^{e-1}$ to $\mathbb{Z}/(p)$ satisfying
\begin{align*}
\psi_{z, W}(x_0, \ldots, x_{e-2})=\begin{cases}
z, & \text{if } (x_0, \ldots, x_{e-2})=(0, \ldots, 0),\\
w\in W, & \text{if } (x_0, \ldots, x_{e-2})\neq (0, \ldots, 0).
\end{cases}
\end{align*}

Given $s$ and $W$, when $W$ is a singleton $\{w\}$, $\psi_{z, W}$ is uniquely determined and equal to $\psi_{z, w}$.
But when $W$ has more than one element, there exist many such $\psi_{z, W}$, especially when the cardinality of $W$ is large.

\begin{Thm}\label{Thm:no-s-uniform}
Assume $g(x_{e-1})$ satisfies \\
{\rm (1)} it is not a permutation polynomial, i.e.,
the image set $I$ of $g(x_{e-1})$ is a proper subset of $\{0, 1, \ldots, p-1\}$,\\
{\rm (2)} for some $r \in I$, there exists $\lambda\neq 0, 1$, for $y\in \mathbb{Z}/(p)$,
such that $g(y)=r$ if and only if $g(\lambda\cdot y)=r$.\\
For an element $s$ in $\mathbb{Z}/(p)$, let $W$ be the nonempty set $\{w\mid s\notin w+I\}$, and
$$\phi(x_0, \ldots, x_{e-1})=g(x_{e-1})+\psi_{z, W}(x_0, \ldots, x_{e-2}).$$
Then we can choose suitable $z$ such that for $\underline{a}, \underline{b}\in G'((x), p^e)$
with $\underline{b}=\lambda\cdot\underline{a}$,
$\phi(\underline{a_{0}}, \underline{a_1},\ldots, \underline{a_{e-1}})$ and
$\phi(\underline{b_{0}}, \underline{b_1},\ldots, \underline{b_{e-1}})$ are $s$-uniform.
\end{Thm}
\begin{proof}
Let $z=s-r$. According to the definition of $W$, for each $w\in W$ and $i\in I$, we have $w+i\neq s$.
Thus if $\phi(x_0, \ldots, x_{e-1})=s$,
then $(x_0, \ldots, x_{e-2}, x_{e-1})=(0, \ldots, 0, y)$ with
$g(y)=r$. By condition ($2$), $g(y)=r$ if and only if $g(\lambda\cdot y)=r$. Thus when
$\underline{b}=\lambda\cdot\underline{a}$,
$\phi(a_0(t), a_1(t), \ldots, a_{e-1}(t))=s$ if and only if  $\phi(b_0(t), b_1(t), \ldots, b_{e-1}(t))=s$.
So they are $s$-uniform.
\end{proof}

We give an example of $g(x_{e-1})$ which satisfies the two conditions in the above theorem.

\begin{Cor}
Assume $g(x_{e-1})$ is not a permutation and $g(y)=g(0)$ if and only if $y=0$. For given $s$
in $\mathbb{Z}/(p)$, let $\phi$ defined as in the above theorem. Then for
$\underline{a}, \underline{b}\in G'(f(x), p^e)$ with $\underline{b}=\lambda\cdot \underline{a}$,
$\phi(\underline{a_{0}}, \underline{a_1},\ldots, \underline{a_{e-1}})$ and
$\phi(\underline{b_{0}}, \underline{b_1},\ldots, \underline{b_{e-1}})$ are $s$-uniform.
\end{Cor}
\begin{proof}
As the preimage of $g(0)$ is $0$. For any $\lambda\ne 0$, the condition ($2$) in the above theorem is satisfied.
The result follows from the above theorem.
\end{proof}

We discuss for a certain $\phi(x_0, \ldots, x_{e-1})$ and different
$\underline{a}, \underline{b}\in G'(f(x), p^e)$,
$\phi(\underline{a_{0}}, \underline{a_1},\ldots, \underline{a_{e-1}})$ and
$\phi(\underline{b_{0}}, \underline{b_1},\ldots, \underline{b_{e-1}})$ can be $s$-uniform for how many $s$.
The following lemma about the sum of Legendre symbols can be found in \cite[Chapter 5]{IR}.

\begin{Lem}
For an odd prime $p$, we have
\begin{align*}
\sum_{x=0}^{p-1}\big(\frac{x^2+w}{p}\big)=\begin{cases}
p-1, & \text{if } p \mid w,\\
-1, & \text{if } p\nmid w.
\end{cases}
\end{align*}
\end{Lem}

\begin{Thm}
Let $g(x_{e-1})=x_{e-1}^2$. Assume
$$\phi(x_0, \ldots, x_{e-1})=g(x_{e-1})+\psi_{0, w}(x_0, \ldots, x_{e-2}).$$
Then for $\underline{a}, \underline{b}\in G'(f(x),p^e)$ with $\underline{b}=-\underline{a}$,
and suitable $w$, there are $[\tfrac{p}{4}]+1$ elements $s$ such that
$\phi(\underline{a_{0}}, \underline{a_1},\ldots, \underline{a_{e-1}})$ and
$\phi(\underline{b_{0}}, \underline{b_1},\ldots, \underline{b_{e-1}})$ are $s$-uniform.
\end{Thm}
\begin{proof}
Let $I$ denote the image set of $x_{e-1}^2$, i.e., $I=\{x^2\mid x\in \mathbb{Z}/(p)\}$, and
$I_{w}=w+I=\{w+x^2\mid x\in \mathbb{Z}/(p)\}$. Then $|I|=|I_{w}|=\tfrac{p+1}{2}$.
As for each $r\in I$, $y^2=r$ if and only if $(-y)^2=r$. By theorem \ref{Thm:no-s-uniform},
we have that for each $s \in I\backslash{I_w}$,
$\phi(\underline{a_{0}}, \underline{a_1},\ldots, \underline{a_{e-1}})$ and
$\phi(\underline{b_{0}}, \underline{b_1},\ldots, \underline{b_{e-1}})$ are $s$-uniform.

Now we count the number of elements in $I\cap I_w$ with $w\neq 0$. We calculate the
sum $1+(\tfrac{w+x^2}{p})$ over all $x\in\mathbb{Z}/(p)$. Let $y=w+x^2$. If $x\neq 0$ and
$y\in I\backslash\{0\}$, the element $y$ is counted $4$ times. If $x=0$ and $y \in I$, $y$
is counted $2$ times. If $y=0$, then $y$ is counted $2$ times.
By the above lemma, then
\begin{align*}
|I\cap I_w| &= \frac{1}{4}\sum_{x\in\ \mathbb{Z}/(p)}(1+\big(\frac{w+x^2}{p}\big))+\frac{1+(\frac{w}{p})}{4}+\frac{1+(\frac{-w}{p})}{4}\\
&=\frac{p-1}{4}+\frac{2+(\frac{w}{p})+(\frac{-w}{p})}{4}\\
&=\frac{p+1+(\frac{w}{p})+(\frac{-w}{p})}{4}
\end{align*}

If $p\equiv 3\mod 4$, then $(\tfrac{w}{p})+(\tfrac{-w}{p})=0$. We have $|I\cap I_w|=\tfrac{p+1}{4}$ and
then $|I\backslash{I_w}|=\tfrac{p+1}{4}=[\tfrac{p}{4}]+1$.

If $p\equiv 1\mod 4$, choose $w$ such that $(\tfrac{w}{p})=(\tfrac{-w}{p})=-1$. Then
$|I\cap I_w|=\tfrac{p-1}{4}$. Thus $|I\backslash{I_w}|=\tfrac{p+3}{4}=[\tfrac{p}{4}]+1$.
The proof is complete.
\end{proof}

\section{Conclusions}
In this article, we consider the distribution properties of compressing sequences derived from
primitive sequences modulo odd prime powers. For strongly primitive polynomial $f(x)$ and compressing map
$$\phi(x_0, x_1, \ldots, x_{e-1})=g(x_{e-1})+\eta_{e-2}(x_0, x_1, \ldots, x_{e-2})$$
with $1\le \deg g\le p-1$, primitive sequences $\underline{a}=\underline{b}$ if and only if the compressing sequences
$\phi(a_0(t), \ldots, a_{e-1}(t))=\phi(b_0(t), \ldots, b_{e-1}(t))$ for all the $t$ with
$\alpha(t)=k$. When $\deg g=1$, we do not need $f(x)$ to be a strongly primitive polynomial. This result
improves the result in \cite[Theorem 5]{TQ}. For $s$-uniform property, when $g(x_{e-1})$
is a permutation polynomial, for a certain $\phi(x_0, \ldots, x_{e-1})$,
the compressing sequences of $\underline{a}$ and $-\underline{a}$ are  $s$-uniform.
When $g(x_{e-1})$ is not a permutation polynomial, there may exist many $\phi(x_0, \ldots, x_{e-1})$ such that
the compressing sequences of $\underline{a}$ and $\lambda\cdot\underline{a}$ are $s$-uniform. For $g(x_{e-1})=x_{e-1}^2$,
we can construct a compressing map $\phi(x_0, \ldots, x_{e-1})$ such that the compressing sequences of $\underline{a}$ and $-\underline{a}$
are $s$-uniform for $[\tfrac{p}{4}]+1$ different $s$ in the image of $\phi$.

\end{document}